\documentclass[12pt,draftcls,onecolumn]{IEEEtran}

\usepackage{amsmath}
\usepackage{amsfonts}
\usepackage{amssymb}

\let\theoremstyle\relax
\usepackage{amsthm}
\usepackage{setspace}
\usepackage{subfigure}
\usepackage{color}
\usepackage{float}
\usepackage{graphicx}
\usepackage{lettrine}
\usepackage{pgfplots}

\usepackage{algorithm}
\usepackage{algpseudocode}

\newtheorem{thm}{Theorem}
\newtheorem{lem}{Lemma}

\usepackage{enumerate}

\theoremstyle{definition}

\newtheorem{defin}{Definition}

\newcommand{\bmat}[1]{\begin{bmatrix}#1\end{bmatrix}}
\newcommand{\bmtx}{\begin{bmatrix}}
\newcommand{\emtx}{\end{bmatrix}}
\newcommand{\bsmtx}{\left[ \begin{smallmatrix}} 
\newcommand{\esmtx}{\end{smallmatrix} \right]}
\newcommand{\bmatarray}[1]{\left[\begin{array}{#1}}
\newcommand{\ematarray}{\end{array}\right]}

\newcommand{\field}[1]{\mathbb{#1}}
\newcommand{\R}{\field{R}}

\newcommand{\C}{\field{C}}

\newcommand{\LtwoT}{\mathcal{L}_2[0,T]}
\newcommand{\Sm}{\field{S}}

\newcommand{\Jdom}{\mathcal{D}}
\newcommand{\Jbnd}{J_{bnd}}
\newcommand{\Jopt}{J^*}
\newcommand{\E}{\mathcal{E}}

\newcommand{\lambdak}{\lambda^{(k)}}
\newcommand{\Lambdak}{\Lambda^{(k)}}
\newcommand{\gk}{g^{(k)}}

\begin{document}

\title{Trajectory-based Robustness Analysis for Nonlinear Systems} 

\author{Peter Seiler\thanks{Email: {\tt\small peter.j.seiler@gmail.com}}
 and Raghu Venkataraman\thanks{Email: {\tt\small veraghu@amazon.com}. }
\thanks{This work was funded by Amazon.com Services LLC.}}

\maketitle  

\begin{abstract}
  This paper considers the robustness of an uncertain nonlinear system
  along a finite-horizon trajectory. The uncertain system is modeled
  as a connection of a nonlinear system and a perturbation. The
  analysis relies on three ingredients.  First, the nonlinear system
  is approximated by a linear time-varying (LTV) system via
  linearization along a trajectory.  This linearization introduces an
  additional forcing input due to the nominal trajectory. Second, the
  input/output behavior of the perturbation is described by
  time-domain, integral quadratic constraints (IQCs). Third, a
  dissipation inequality is formulated to bound the worst-case
  deviation of an output signal due to the uncertainty. These steps
  yield a differential linear matrix inequality (DLMI) condition to
  bound the worst-case performance.  The robustness condition is then
  converted to an equivalent condition in terms of a Riccati
  Differential Equation.  This yields a computational method that
  avoids heuristics often used to solve DLMIs, e.g. time gridding. The
  approach is demonstrated by a two-link robotic arm example.
\end{abstract}

\section{Introduction}

This paper develops theoretical and computational methods to analyze
the robustness of uncertain nonlinear systems over finite time
horizons.  Motivating applications include robotic systems
\cite{murray94}, space launch vehicles \cite{marcos09,biertumpfel21},
and aircraft during the landing phase \cite{theis18}. The analysis in
this paper considers an uncertain system modeled by an
interconnection of a (nominal) nonlinear, time-varying system and an
uncertainty.  The uncertainty can model dynamic or parametric
uncertainty. It can also model non-differentiable nonlinearities,
e.g. saturation or deadzone. The input-output properties of the
uncertainty are characterized by integral quadratic constraints (IQCs)
\cite{megretski97,veenman16}.  The main result in \cite{megretski97}
is an IQC stability theorem evaluated on an infinite time-horizon for
the case where the nominal system is linear and time-invariant (LTI).

This paper makes two contributions. The first contribution is a
theoretical condition to compute finite horizon robustness metrics for
uncertain nonlinear systems. The objective is to assess the worst-case
deviation from a nominal trajectory caused by the uncertainty.  This
analysis problem is approximated, via linearization along the
trajectory, by the following problem: compute the worst-case
$\mathcal{L}_2$-norm for an output of an uncertain LTV system over the
set of uncertainties (Section~\ref{sec:probform}).  The uncertain LTV
system includes a forcing due to the nominal trajectory.  This forcing
is addressed by using an augmented LTV system with a non-zero initial
condition as done previously in \cite{schweidel20,biertumpfel23}. The
main technical result (Theorem~\ref{thm:RobBound} in
Section~\ref{sec:wcnorm}) is a differential linear matrix inequality
(DLMI) condition to bound the worst-case deviation.  This theorem uses
standard dissipation inequality and IQCs results.

The ``tightest'' bound on worst-case deviation can be computed by a convex
optimization with DLMI constraints on a storage function matrix $P(t)$
and IQC variables. This involves infinite dimensional constraints due
to dependence on $t$ and a search over the space of differentiable
functions for $P(t)$.  A typical heuristic is to enforce the
constraints on a finite time grid and use a finite-dimensional
parameterization of $P(t)$ using basis functions.  This heuristic has
been used for both LTV
\cite{moore15,seiler17arXiv,seiler19,biertumpfel23} and linear
parameter varying (LPV) systems \cite{pfifer16IJRNC,wu96}. However
there are no formal guarantees, in general, with these approximations.

A second contribution of this paper is a computational algorithm to
assess the worst-case deviation that avoids these heuristics
(Section~\ref{sec:compalg}).  We instead convert to an equivalent
finite dimensional optimization.  This step relies on a connection to
Riccati Differential Equations (RDEs) using the LTV Bounded Real Lemma
\cite{tadmor90,ravi91,green95,chen00,basar08}.  This connection can
also be used to compute subgradients with minimal computational cost
(Section~\ref{sec:subgrad}). This builds on related prior work by
\cite{jonsson02,petersen00,kao01}. As a result the ellipsoid algorithm
can be used to solve the optimization to within a desired accuracy
(Section~\ref{sec:ellipalg}). This is efficient if the IQC is
parameterized by a small number of variables.  The approach is
demonstrated by analyzing the robustness of a two-link robot tracking
a desired trajectory (Section~\ref{sec:ex}).

Finite-horizon robustness of continuous-time LTV systems has also been
considered in \cite{jonsson02,petersen00} and more recently in
\cite{moore15,seiler17arXiv,seiler19,biertumpfel21}. There is also
related work in discrete-time, e.g. \cite{abou22,fry17,palframan17}.
These works mainly treat the uncertain LTV system as the starting
point for the analysis. One exception is \cite{biertumpfel23} which
considers the effect of parametric uncertainty in a nonlinear
model. The work in \cite{biertumpfel23} also uses linearization along
a nominal trajectory. It uses the basis function/time gridding
heuristic to compute bounds on the induced $\mathcal{L}_2$ gain of the
uncertain LTV approximation from an exogenous disturbance to an output
signal.  Our paper is similar but we use a different formulation for
the robustness metric and our computational algorithm avoids the use
of heuristic time gridding and basis functions.

\textbf{Notation:} $\mathbb{R}^n$ and $\C^n$ denote the sets of
$n$-by-$1$ real and complex vectors.  $\mathbb{R}^{n \times m}$ and
$\Sm^{n}$ denote the sets of $n$-by-$m$ real matrices and $n$-by-$n$
real, symmetric matrices. The $\mathcal{L}_2^n[0,T]$ norm of a signal
$v:[0,T] \rightarrow \R^n$ on a finite horizon $T<\infty$ is
$\|v\|_{2, [0,T]} := \left( \int_0^T v(t)^\top v(t) dt \right)^{1/2}$.
If $\|v\|_{2, [0,T]} < \infty$ then $v \in \LtwoT$. The superscript $n$ 
is omitted from $\mathcal{L}_2^n[0,T]$ when the dimension is clear.



\section{Problem Formulation}
\label{sec:probform}

\subsection{Uncertain Nonlinear System}

Consider an uncertain system defined by the interconnection of a
time-varying, nonlinear system $G_{NL}$ and an uncertainty $\Delta$ as
shown in Figure~\ref{fig:GNLunc}. The system $G_{NL}$ is described by
the following state-space model:
\begin{equation}
  \label{eq:NLnom}
  \begin{split}
    \dot{x}(t) & = f(x(t),w(t),d(t),t) \\
    v(t) & = g_v(x(t),w(t),d(t),t) \\
    e(t) & = g_e(x(t),w(t),d(t),t) 
  \end{split}
\end{equation}
where $x(t) \in \R^{n_x}$ is the state at time $t$. The inputs of
$G_{NL}$ at time $t$ are $w(t)\in \R^{n_w}$ and $d(t)\in \R^{n_d}$
while $v(t)\in \R^{n_v}$ and $e(t)\in \R^{n_e}$ are outputs.  The
vector field $f:\R^{n_x \times n_w \times n_d \times 1}\to \R^{n_x}$
is assumed to be continuously differentiable.  Similarly, the output
mappings $g_v$ and $g_e$ are assumed to be continuously
differentiable. 


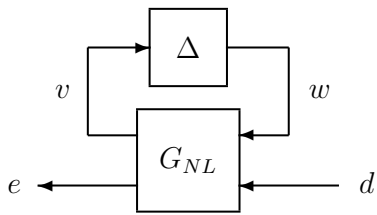
\begin{figure}[h!]
\centering
\scalebox{0.95}{
\begin{picture}(172,90)(23,20)
 \thicklines
 \put(75,25){\framebox(40,40){$G_{NL}$}}
 \put(163,32){$d$}
 \put(155,35){\vector(-1,0){40}}  
 \put(23,32){$e$}
 \put(75,35){\vector(-1,0){40}}  
 \put(80,75){\framebox(30,30){$\Delta$}}
 \put(42,70){$v$}
 \put(55,55){\line(1,0){20}}  
 \put(55,55){\line(0,1){35}}  
 \put(55,90){\vector(1,0){25}}  
 \put(143,70){$w$}
 \put(135,90){\line(-1,0){25}}  
 \put(135,55){\line(0,1){35}}  
 \put(135,55){\vector(-1,0){20}}  
\end{picture}
} 
\vspace{-.3cm}
\caption{Uncertain system defined by the interconnection of a
  time-varying nonlinear system $G_{NL}$ and uncertainty $\Delta$.}
\label{fig:GNLunc}
\end{figure}


The uncertainty is a causal operator
$\Delta:\mathcal{L}_{2}^{n_v}[0,T] \rightarrow
\mathcal{L}_{2}^{n_w}[0,T]$ that maps $v$ to $w$. It is assumed to be
an element of a set $\mathbf{\Delta}$ of block-structured
uncertainties as is standard in robust control \cite{zhou96}.  The
uncertainty $\Delta\in \mathbf{\Delta}$ can include blocks for
unmodeled dynamics, parametric variations, and/or infinite
dimensional operators (e.g. time delays). It can also include
non-differentiable nonlinearities, e.g. saturations, that are separate
from the differentiable nonlinearities in $G_{NL}$.  The term
``uncertainty'' is used for simplicity when referring to $\Delta$.

Next, assume an input $\bar{d}:[0,T]\to \R^{n_d}$ and initial condition
$x(0)=\bar{x}_0$ are given.  The nominal trajectory of the uncertain
system is obtained when $\Delta=0$. This yields $\bar{w}(t)=0$ while
the other nominal signals $(\bar{x},\bar{v},\bar{e})$, satisfy:
\begin{align}
  \begin{split}
    \dot{\bar x}(t) & = f(\bar x(t),0,\bar d(t)) \\
    \bar v(t) & = g_v(\bar x(t),0,\bar d(t)) \\
    \bar e(t) & = g_e(\bar x(t),0, \bar d(t)).
  \end{split}
\end{align}
If $\Delta \ne 0$ then $w\ne 0$, in general. This will perturb the
uncertain system from the nominal trajectory giving new signals
$(x,v,e)$. These perturbed signals depend on the specific uncertainty
$\Delta \in \mathbf{\Delta}$. We assume that the nominal solution
exists on $[0,T]$ and also that the perturbed solutions exist on
$[0,T]$ for all $\Delta \in \mathbf{\Delta}$. The objective is to
bound (approximately) the worst-case deviation of the output signal
in the $\mathcal{L}_2[0,T]$ norm:
$\max_{\Delta \in \mathbf{\Delta}} \| e - \bar{e} \|_{2,[0,T]}$.

\subsection{Linearization Along Trajectory}
\label{sec:trajlin}

The approach taken here is to linearize the dynamics of $G_{NL}$
around the nominal trajectory (assuming $\bar{d}$ and $\bar{x}_0$ are
fixed):
\begin{align*}
  \begin{split}
    \dot{x}(t) & = \dot{\bar x}(t) 
       + A(t) \, (x(t)-\bar{x}(t)) + B(t) \, w(t) \\
    v(t) & = \bar v(t) + C_v(t) \, (x(t)-\bar{x}(t)) 
       + D_{vw}(t) \, w(t) \\
    e(t) & = \bar e(t) + C_e(t) \, (x(t)-\bar{x}(t)) 
       + D_{ew}(t) \, w(t).
  \end{split}
\end{align*}
The time-varying matrices are given by gradients evaluated along the
nominal trajectory, e.g.
$A(t):=\nabla_x f|_{(\bar{x}(t),0,\bar{d}(t))}$.  The linearization
can be re-written in perturbation coordinates: $\delta_x:=x-\bar{x}$,
$\delta_v:=v-\bar{v}$, and $\delta_e:=e-\bar{e}$. This yields a
linear time-varying (LTV) approximation, $G_{LTV}$:
\begin{align}
  \label{eq:Gltv}
  \begin{split}
    \dot{\delta}_x(t) & = 
       A(t) \, \delta_x(t) + B(t) \, w(t) \\
    \delta_v(t) & = C_v(t) \, \delta_x(t)
       + D_{vw}(t) \, w(t) \\
    \delta_e(t) & = C_e(t) \, \delta_x(t)
       + D_{ew}(t) \, w(t).
  \end{split}
\end{align}
The uncertainty does not effect the initial condition, i.e. the
perturbed and nominal trajectories have the same initial condition
$x(0)=\bar{x}_0$. Hence the initial condition of $G_{LTV}$ is
$\delta_x(0)=0$.

Figure~\ref{fig:GLTVunc} shows the linearized approximation for the
original uncertain system. The nonlinear system $G_{NL}$ is replaced
by its linearization $G_{LTV}$.  The nominal trajectory affects the
linearized model in two ways. First, the state-matrices of
\eqref{eq:Gltv} depend on $(\bar{x},\bar{d})$. Second, the input to
$\Delta$ is $v=\bar{v}+\delta_v$ and includes forcing due to nominal
input $\bar{v}$.  If $\Delta \ne 0$ then the nominal signal $\bar{v}$
will force $G_{LTV}$ via $w$ thus generating a perturbed output
$\delta_e \ne 0$.


\begin{figure}[h!]
\centering
\scalebox{0.95}{
\begin{picture}(125,85)(20,20)
 \thicklines
 \put(75,25){\framebox(40,40){$G_{LTV}$}}
 \put(24,31){$\delta_e$}
 \put(75,35){\vector(-1,0){40}}  
 \put(42,55){$\delta_v$}
 \put(75,55){\line(-1,0){20}}  
 \put(55,55){\vector(0,1){32}}  
 \put(23,87){$\bar{v}$}
 \put(32,90){\vector(1,0){20}}  
 \put(55,90){\circle{6}}  
 \put(58,90){\vector(1,0){22}}  
 \put(65,95){$v$}
 \put(80,75){\framebox(30,30){$\Delta$}}
 \put(141,70){$w$}
 \put(135,90){\line(-1,0){25}}  
 \put(135,55){\line(0,1){35}}  
 \put(135,55){\vector(-1,0){20}}  
\end{picture}
} 
\vspace{-.3cm}
\caption{Uncertain system defined by the interconnection of an LTV
  system $G_{LTV}$ and uncertainty $\Delta$ including forcing due to
  $\bar{v}$.}
\label{fig:GLTVunc}
\end{figure}
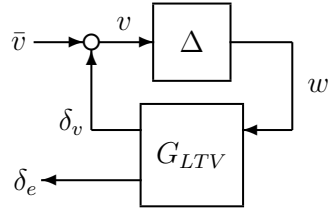

The precise problem addressed by this paper is to compute a bound on
the worst-case deviation for the uncertain system in
Figure~\ref{fig:GLTVunc}:
\begin{align}
  \max_{\Delta \in \mathbf{\Delta}} \| \delta_e \|_{2,[0,T]} 
\end{align}
The analysis is based on the linearized model $G_{LTV}$. Hence it
assumes that $w$ is sufficiently small that the higher-order terms
dropped in the linearization are negligible.  It is possible to bound
the effect of linearization errors as in
\cite{takarics15,biertumpfel23} but we will not do so here. Another
assumption, implicit in Figure~\ref{fig:GNLunc}, is that the
uncertainty $\Delta$ enters in a rational, i.e. feedback, form. This
is sufficient for many types of unmodeled dynamics, delays, and
non-differentiable nonlinearities.  However, parametric uncertainty
often appears in a non-rational form in a nonlinear model.  Parametric
uncertainties are treated in \cite{biertumpfel23} via linearization
along a nominal trajectory. The approach in \cite{biertumpfel23} could
be combined with the method in this paper but, again, this is not
pursued.

\section{Worst-Case Norm}
\label{sec:wcnorm}

\subsection{Augmented System}

We'll focus on the worst-case analysis problem formulated in
Section~\ref{sec:trajlin} with the linearized dynamics.  The first
step is to incorporate the effect of $\bar{v}$ into an augmented LTV
system with state $x_a:= \bsmtx \delta_x \\ 1 \esmtx \in \R^{n_x+1}$.
The linearized approximation in Figure~\ref{fig:GLTVunc} is equivalent
to the interconnection of $w=\Delta(v)$ and the following LTV system:
\begin{align}
  \label{eq:Ga}
  \begin{split}
    \dot{x}_a(t) & =  A_a(t) \, x_a(t) + B_a(t) \, w(t) \\
    v(t) & = C_{v,a}(t) \, x_a(t) + D_{vw}(t) \, w(t) \\
    \delta_e(t) & = C_{e,a}(t) \, x_a(t) + D_{ew}(t) \, w(t) \\
      x_a(0) & = \bsmtx 0 \\ 1 \esmtx,
  \end{split}
\end{align}
where the augmented state matrices are defined as:
\begin{align*}
& A_a(t) := \bmtx A(t) & 0 \\ 0 & 0 \emtx,
\,\,
B_a(t):=  \bmtx B(t) \\ 0 \emtx, \\
& C_{v,a}(t):= \bmtx C_v(t) & \bar{v} \emtx,
\,\,
C_{e,a}(t) := \bmtx C_e(t) & 0 \emtx.
\end{align*}

The augmented system \eqref{eq:Ga}, denoted $G_a$, has a non-zero
initial condition. The definition of $A_a$ and $B_a$ ensure that the
last (scalar) entry of $x_a$ is identically equal to 1 for all
solutions. Thus the output equation for $v(t)$ includes the effect of
the nominal signal $\bar{v}$. This augmented state method was used in
\cite{schweidel20,biertumpfel23} for similar trajectory-based
linearizations.  The offset $\bar{v}$ is needed to ensure that
the linearization provides an accurate approximation.

One technical issue is that the uncertain LTV system could be
ill-posed.  Specifically, $w=\Delta(v)$ and the output equation
$v=C_{v,a}x_a+D_{vw}w$ may involve an algebraic equation. For example,
if the uncertainty is a time-varying matrix $\Delta(t)$ then the
algebraic equation at each time is:
\begin{align}
   (I-D_{vw}(t)\Delta(t)) \, v(t) = C_{v,a}(t) \, x_a(t)
\end{align}
This algebraic equation will have no solutions or non-unique solutions
if $(I-D_{vw}(t)\Delta(t))$ is singular. Such ill-posed
cases can also occur when $\Delta$ is not necessarily a time-varying
matrix. The uncertain system is said to be well-posed if such cases do not
occur as formally defined next.

\begin{defin}
\label{def:wellposed}
Consider the uncertain system defined by the interconnection of $G_a$
in \eqref{eq:Ga} and the set of causal uncertainties
$\mathbf{\Delta}$. The uncertain system is well-posed if for each
$\Delta \in \mathbf{\Delta}$ there exists unique solutions
$(x_a,v,w,\delta_e) \in \mathcal{L}_2^{n_x+1}[0,T]$ satisfying
Equation~\eqref{eq:Ga} with $x_a(0)=\bsmtx 0 \\ 1 \esmtx$ and
$w=\Delta(v)$.
\end{defin}



\subsection{Integral Quadratic Constraints (IQCs)}

The next step is to bound the input/output behavior of
$\Delta\in \mathbf{\Delta}$.  We will use the class of time-domain
Integral Quadratic Constraints (IQCs) \cite{megretski97} defined
below for this step.

\begin{defin}
  \label{def:tdiqc}
  A causal operator
  $\Delta:\mathcal{L}_2^{n_v}[0,T] \rightarrow
  \mathcal{L}_2^{n_w}[0,T]$ satisfies the time-domain IQC
  defined by $M\in \Sm^{n_v+n_w}$ if the following inequality holds
  $\forall v \in \mathcal{L}_2^{n_v}[0,T]$ and $w=\Delta(v)$:
  \begin{align}
    \label{eq:tdiqc}
    \int_0^T \bmtx v(t) \\ w(t) \emtx^\top M 
        \bmtx v(t) \\ w(t) \emtx  \, dt \, \ge 0.
  \end{align}
\end{defin}
\vspace{0.05in} 



Time-domain IQCs can be used to bound the effect of various
uncertainties \cite{megretski97}.  For example, consider an
uncertainty that is norm-bounded in the induced $\mathcal{L}_2$ norm:
$\|\Delta\|_{2\to 2}\le \beta$ for some $\beta<\infty$. This
uncertainty satisfies the IQC defined by
$M = \bsmtx \beta^2 I & 0 \\ 0 & -I \esmtx$. As another example, a
memoryless nonlinearity in the sector $[\alpha,\beta]$ satisfies the
IQC defined by
$M = \bsmtx -2\alpha\beta & \alpha+\beta \\ \alpha+\beta & -2 \esmtx$.
An important point is that IQCs can be combined to form new IQCs.
Suppose $\Delta$ satisfies the IQCs defined by
$\{ M_i \}_{i=1}^{m}\subset \Sm^{n_v+n_w}$. Then $\Delta$ also
satisfies the IQC defined by $M(\lambda):=\sum_{i=1}^m \lambda_i M_i$
for any $\lambda_i \ge 0$.

Time domain IQCs, as defined above, are a special case of more general
(and powerful) IQCs given in the literature. The integrand in
\eqref{eq:tdiqc} is a quadratic function of the input/output signals
$(v,w)$. These are called non-dynamic IQCs.  This is in contrast with
dynamic IQCs that express the integrand as a quadratic function of
filtered signals of $(v,w)$. Moreover, Definition~\ref{def:tdiqc}
requires the constraint to hold over the finite time horizon
$T>0$. These are often referred to as hard IQCs
\cite{megretski97}. This is in contrast to soft IQCs that only hold,
in general, on an infinite horizon. The algorithm in this paper can be
adapted to handle dynamic, hard IQCs with mainly notational changes.



\subsection{Condition to Bound the Worst-Case Deviation}

The next theorem gives a condition to bound the worst-case
deviation $\max_{\Delta \in \mathbf{\Delta}} \|\delta_e\|_{2,[0,T]}$.
The proof uses IQCs and a standard dissipation
argument~\cite{schaft99, willems72a, willems72b, khalil01}.

\begin{thm}
\label{thm:RobBound}
Assume the uncertain system defined by the interconnection of $G_a$ in
\eqref{eq:Ga} and the set of causal, uncertainties $\mathbf{\Delta}$
is well-posed. Furthermore, assume each $\Delta \in \mathbf{\Delta}$
satisfies the IQCs defined by
$\{ M_i \}_{i=1}^{m}\subset \Sm^{n_v+n_w}$.

If there exist non-negative scalars $\{\lambda_i\}_{i=1}^m$ and a
differentiable function $P:[0,T]\to \Sm^{n_x+1}$ such that
$P(T)\succeq 0$ and\footnote{The notation $(\cdot)^\top$ in
  \eqref{eq:DLMI} corresponds to a factor than can be determined
  from symmetry and hence is omitted.}
\begin{align}
\label{eq:DLMI}
&\bsmtx \dot{P} + A_a^\top P + P A_a  &  P B_a  \\    B_a^\top P & 0  \esmtx
+ (\cdot)^\top 
\bsmtx C_{a,e} & D_{ew} \esmtx  \\
\nonumber
& 
+ \sum_{i=1}^m \lambda_i \,
(\cdot)^\top 
M_i \bsmtx C_{a,v} & D_{vw} \\ 0 & I \esmtx  
\prec 0 \hspace{0.2in} \forall t\in [0,T]
\end{align}
then
\begin{align} 
  \label{eq:wcnorm}
  \max_{\Delta \in \mathbf{\Delta}} \|\delta_e\|_{2,[0,T]} \le
  \left[ \bsmtx 0\\ 1\esmtx^\top P(0) \bsmtx 0\\ 1 \esmtx \right]^{\frac{1}{2}}.
\end{align}
\end{thm}
\begin{proof}
  Define a storage function $V:\R^{n_x+1}\times\R\to \R$ by
  $V(x_a,t) := x_a^\top P(t) x_a$.  Consider any
  $\Delta \in \mathbf{\Delta}$.  By well-posedness, the uncertain
  system with $\Delta$ and $x_a(0)=\bsmtx 0 \\ 1 \esmtx$ has a unique
  solution $(x_a,v,w,\delta_e)$.  Left and right multiply
  \eqref{eq:DLMI} by $[x_a^\top, w^\top]$ and its transpose to show that
  $V$ satisfies the dissipation inequality 
  $\forall t\in [0,T]$:
  \begin{align*}
    \dot{V} + \delta_e^\top \delta_e 
    + \sum_{i=1}^m \lambda_i  \bsmtx v \\ w \esmtx^\top M_i \bsmtx v \\ w \esmtx
    \le 0 
  \end{align*}
  Integrate over $[0,T]$ to obtain:
  \begin{align*}
    &  V(x_a(T),T) - V(x_a(0),0) 
     + \|\delta_e\|^2_{2,[0,T]} \\
    & + \sum_{i=1}^m \lambda_i \int_{0}^{T}  \bmtx v(t) \\ w(t) \emtx^\top 
              M_i  \bmtx v(t) \\ w(t) \emtx dt 
    \le 0.
  \end{align*}
  Apply $P(T) \succeq 0$, $x_a(0)=\bsmtx 0 \\ 1 \esmtx$, and the IQCs
  defined by $\{M_i\}_{i=1}^m$ to conclude:
  \begin{align}
    \|\delta_e\|_{2,[0,T]}^2 
   \le \bsmtx 0\\ 1\esmtx^\top P(0) \bsmtx 0\\ 1\esmtx.
  \end{align}
  This inequality holds for all $\Delta\in\mathbf{\Delta}$ and hence this
  yields the  bound in \eqref{eq:wcnorm}.
\end{proof}

The inequality in Equation~\ref{eq:DLMI} is compactly denoted as
$DLMI(t,P,\lambda)\prec 0$.  This notation emphasizes that the
constraint is a time-dependent, differential linear matrix inequality
(DLMI) in $(P,\lambda)$.  The dependence on the state matrices of
$G_a$ and the IQC matrices $\{M_i\}_{i=1}^n$ is not explicitly denoted
but will be clear from context. 

The tightest upper bound on the worst-case deviation, based on
Theorem~\ref{eq:DLMI}, is obtained by solving the following
optimization:
\begin{align}
\label{eq:DLMIOptim}
\Jopt = &  \min_{\lambda\ge 0, \, P} \, \bsmtx 0\\ 1\esmtx^\top P(0) \bsmtx 0\\ 1\esmtx \\
\nonumber
&  \mbox{subject to:} \,\,\, P(T) \succeq 0, \\
\nonumber
&    DLMI(t,P,\lambda) \prec 0 \,\,\, \forall t \in[0,T] 
\end{align}
The worst-case deviation is upper bounded by $\sqrt{J_{bnd}}$.
The optimization involves convex constraints on the optimization
variables $\lambda$ and $P$. Moreover, the cost is a linear (and hence
convex) function of $P$. Thus Equation~\ref{eq:DLMIOptim} is a
(convex) semidefinite program (SDP). However, there are two main issues with
solving this optimization. First, the DLMI corresponds to an infinite
number of constraints since it must hold for all $t\in [0,T]$.
Second, the optimization requires a search over the space of
differentiable functions $P:[0,T]\to \Sm^{n_x+1}$.  A heuristic
approach to approximately solve this optimization involves
\cite{moore15,seiler19,biertumpfel23}: (i) enforcing the DLMI on a
finite time grid, and (ii) restricting $P$ to a linear combination of
differentiable basis functions. These approximations yield a finite
dimensional optimization but provides no guarantees on the solution
accuracy.

\section{Computational Algorithm}
\label{sec:compalg}

This section presents a computational method to convert the
optimization \eqref{eq:DLMIOptim} to an equivalent finite-dimensional
optimization. This enables solutions via cutting plane methods without
resorting to time-griding or basis functions.

\subsection{Finite-Dimensional Optimization}

The first step is to define a function $J$ involving the minimization
over $P$ for a fixed $\lambda \in \R^m$:
\begin{align}
\label{eq:defJ}
J(\lambda) :=  & 
\min_{P} \bsmtx 0\\ 1\esmtx^\top P(0) \bsmtx 0\\ 1\esmtx \\
\nonumber
&  \mbox{subject to:} \,\,\, P(T) \succeq 0, \\
\nonumber
&    DLMI(t,P,\lambda) \prec 0 \,\,\, \forall t \in[0,T]
\end{align}
If the optimization is infeasible for a given $\lambda$ then
$J(\lambda)=+\infty$.  Define the domain of $J$ as
$\Jdom:=\{ \lambda \in \R^m \, : \, J(\lambda)<\infty \}$.  It follows
from the linear constraints and cost of \eqref{eq:defJ} that
$J:\Jdom\to \R$ is a convex function and the domain $\Jdom$ is a
convex set.  The infinite dimensional optimization
\eqref{eq:DLMIOptim} can re-written in terms of $J$ as follows:
\begin{align}
  \label{eq:minJ}
\Jopt =  & \min_{\lambda \in \Jdom ,\lambda\ge 0} J(\lambda) 
\end{align}
Equation~\ref{eq:minJ} is, formally, a finite-dimensional convex
optimization with decision variables $\lambda\in \R^m$. However, evaluating
definition of $J$ in \eqref{eq:defJ} still involves a minimization over
$P$ with the time-dependent DLMI constraint. 

Next, we show that $J(\lambda)$ can be evaluated directly without
explicitly performing the minimization over $P$.  To simplify
notation, consider the DLMI in \eqref{eq:DLMI} and define $(Q,S,R)$ as
follows:
\begin{align}
\label{eq:QSR}
& \bmtx Q & S \\ S^\top & R \emtx :=
\bmtx Q_0 & S_0 \\ S_0^\top & R_0 \emtx
+ \sum_{i=1}^m \lambda_i \, \bmtx Q_i & S_i \\ S_i^\top & R_i \emtx
\end{align}
where:
\begin{align}
\label{eq:QSRi}
\begin{split}
\bsmtx Q_0 & S_0 \\ S_0^\top & R_0 \esmtx 
& := \bsmtx C_{a,e}^\top \\ D_{ew}^\top \esmtx 
      \bsmtx C_{a,e} & D_{ew} \esmtx  \\
\bsmtx Q_i & S_i \\ S_i^\top & R_i \esmtx 
  & := \bsmtx C_{a,v} & D_{vw} \\ 0 & I \esmtx^\top
 M_i \bsmtx C_{a,v} & D_{vw} \\ 0 & I \esmtx.
\end{split}
\end{align}
The matrices are defined by the appropriate block partitioning.  Here
$(Q,S,R)$ are functions of $(t,\lambda)$ and
$\{ (Q_i,S_i,R_i) \}_{i=1}^m$ are only functions of $t$. The DLMI in
\eqref{eq:DLMI} can thus be expressed as:
\begin{align}
\label{eq:DLMIqsr}
\hspace{-0.1in}
\bmtx \dot{P} + A_a^\top P + P A_a  &  P B_a  \\    B_a^\top P & 0  \emtx
+ \bmtx Q & S \\ S^\top & R \emtx \prec 0 
\end{align}
The matrices in the DLMI can be used to define a related
Riccati Differential Equation (RDE):
\begin{align}
\label{eq:RDE}
\begin{split}
& \dot{Y} + A_a^\top Y + Y A_a + Q \\
& \hspace{0.2in} - (Y B_a + S) R^{-1} (Y B_a + S)^\top = 0
\end{split}
\end{align}
This is compactly denoted as $RDE(t,Y,\lambda)=0$. The next
theorem states that $J(\lambda)$ can be evaluated from the solution to
this RDE.

\begin{thm}
  \label{thm:Jeval}
  Assume $(A_a,B_a,Q,R,S)$ are all continuous functions of time.
  Moreover, assume $\lambda \in\R^m$ is given and $R(t,\lambda) <0$
  for all $t \in [0,T]$.  Then the following are equivalent:
  \begin{enumerate}
  \item $\lambda \in \Jdom$, i.e. $J(\lambda) < \infty$.
  \item There exists a differentiable function
    $P:[0,T]\to \Sm^{n_x+1}$ that satisfies $P(T)\succeq 0$ and
    $DLMI(t,P,\lambda)\prec 0$ for all $t \in [0,T]$.
  \item There exists a differentiable function
    $Y:[0,T]\to \Sm^{n_x+1}$ that satisfies $Y(T)=0$ and 
    $RDE(t,Y,\lambda) = 0$ for all $t \in [0,T]$.
  \end{enumerate}
  Moreover, if the conditions hold then:
  \begin{align}
    \label{eq:JfromY}
    J(\lambda)=\bsmtx 0 \\ 1\esmtx^\top Y(0) \bsmtx 0 \\ 1 \esmtx.
  \end{align}
\end{thm}
\begin{proof}
  The equivalence of 1 and 2 is a consequence of the definition of $J$
  in \eqref{eq:defJ}. The remainder of the proof shows that 2 and 3
  are equivalent.  

  Condition 2 holds if and only if there exists $\epsilon_1>0$ such
  that the following Riccati Differential Inequality (RDI) holds
  for all $t\in [0,T]$:
  \begin{align}
    \label{eq:RDI}
    \begin{split}
    & \dot{P} + A_a^\top P + P A_a + Q \\
    & - (P B_a + S) R^{-1} (P B_a + S)^\top  + \epsilon_1 \cdot I\prec 0 
    \end{split}
  \end{align}
  This follows from the Schur complement lemma \cite{boyd94} and
  $R(t,\lambda) <0$.\footnote{The interval $[0,T]$ is compact since
    $T<\infty$. Hence a strict matrix inequality $M(t)<0$
    $\forall t\in [0,T]$ holds if and only if $\exists \epsilon_1>0$
    such that $M(t)+\epsilon_1 \cdot I <0$ $\forall t\in [0,T]$.}  The
  Bounded Real Lemma for LTV systems
  \cite{tadmor90,ravi91,green95,chen00,basar08} states that there
  exists a differentiable function $P$ satisfying $P(T)\succeq 0$ and
  \eqref{eq:RDI} if and only if Condition 3 holds.  The precise
  version of the LTV Bounded Real Lemma used here is Theorem 1 in
  \cite{seiler17arXiv,seiler19}.



  To conclude the proof we assume the conditions hold and show
  \eqref{eq:JfromY} is true.  As noted above, if $P$ satisfies the
  DLMI then $P$ satisfies the RDI in \eqref{eq:RDI}. Thus there exists
  $W:[0,T]\to \Sm^{n_x+1}$ such that $W(t)\prec 0$ for all
  $t\in [0,T]$ and:
  \begin{align}
    \label{eq:RDEW}
    & \dot{P} + A_a^\top P + P A_a + Q \\
    \nonumber
    & - (P B_a + S) R^{-1} (P B_a + S)^\top = W
      \hspace{0.1in} \forall t \in [0,T]
  \end{align}
  This is an RDE with a perturbation $W$ on the right side.  We denote
  \eqref{eq:RDEW} by $RDE(t,P,\lambda)=W$.  It follows from
  Lemma~\ref{lemApp:RDEmon} in the Appendix that $P(0)\succeq Y(0)$.
  This inequality holds for any $P$ that satisfies Condition 2 so that
  $J(\lambda) \ge \bsmtx 0 \\ 1 \esmtx^\top Y(0) \bsmtx 0 \\ 1
  \esmtx$.

  Next, if $P$ satisfies Condition 2 then it satisfies the RDI in
  \eqref{eq:RDI}.  This implies (via the equivalence of Conditions 2
  and 3 to this perturbed RDI) that there exists a differentiable
  function $Y_1:[0,T]\to \Sm^{n_x+1}$ that satisfies $Y_1(T)= 0$ and
  $RDE(t,Y_1,\lambda)=-\epsilon_1 \cdot I$. Define
  $\epsilon_k=\frac{1}{k}\epsilon_1$ for $k=2,3,\ldots$ and let
  $Y_k:[0,T]\to \Sm^{n_x+1}$ be the solution to $Y_k(T)= 0$ and
  $RDE(t,Y_k,\lambda)=-\epsilon_k \cdot I$. It follows from
  Lemma~\ref{lemApp:RDEcont} in the Appendix that
  $\lim_{k\to \infty} \|Y_k(0) - Y(0)\|=0$.  Moreover, each
  $\{Y_k\}_{k=1}^\infty$ satisfies the DLMI by the Schur complement
  lemma.  Thus the optimization in \eqref{eq:defJ} has feasible points
  arbitrarily close to $Y(0)$ and hence
  $J(\lambda) = \bsmtx 0 \\ 1 \esmtx^\top Y(0) \bsmtx 0 \\ 1 \esmtx$.
\end{proof}



\subsection{Subgradients}
\label{sec:subgrad}

By Theorem~\ref{thm:Jeval}, $J$ can be evaluated directly from the
solution of a RDE.  There is no need to approximate the optimization
by enforcing the DLMI on a time grid and using bases functions for
$P$.  This section demonstrates that subgradients on $J$ can be
computed with small additional computation. The subgradients are
evaluated based on a linear quadratic (LQ) optimization defined in
terms of $(Q,S,R)$ from \eqref{eq:QSR}:
\begin{align}
\nonumber
  & \max_{w\in \mathcal{L}_2[0,T]} 
  \int_0^T \bsmtx x_a(t) \\ w(t) \esmtx^\top
  \bsmtx Q(t,\lambda) & S(t,\lambda) \\ 
  S(t,\lambda)^\top & R(t,\lambda) \esmtx
  \bsmtx x_a(t) \\ w(t) \esmtx \, dt \\
\label{eq:LQcost}
  & \mbox{subject to:} \,\,
  \dot{x}_a = A_a \, x_a + B_a \, w, \,\,
   x_a(0)=\bsmtx 0 \\ 1 \esmtx
\end{align}
The next theorem is a variation of the cutting plane results in
\cite{jonsson02,kao01}.

\begin{thm}
\label{thm:Jsubgrad}
If $\lambda \in \Jdom$ then:
\begin{enumerate}
\item The optimal cost for \eqref{eq:LQcost} is equal to
  $J(\lambda)$.
\item The optimal cost is achieved by $(x_a^*,w^*)$ satisfying the
  following with $x_a^*(0) =\bsmtx 0 \\ 1 \esmtx$:
  \begin{align}
    \nonumber
    \dot{x}_a^* & := \left(A_a-B_a R^{-1} \left( Y B_a +S \right)^\top \right)
                  x_a \\
    \label{eq:xawstar}
    w^* & := -R^{-1} \left( Y B_a +S \right)^\top x_a^*,
  \end{align}
  where $Y$ is the solution to $RDE(t,Y,\lambda)=0$ with $Y(T)=0$.
\item $J$ satisfies the subgradient inequality:
  \begin{align}
    \label{eq:Jsubgrad}
    J(\alpha) \ge J(\lambda) + g^\top (\alpha-\lambda)
    \,\,\, \forall \alpha \in \R^m
  \end{align}
  where $g\in \R^m$ is defined by: 
  \begin{align}
    \label{eq:gi}
    g_i :=   \int_0^T \bsmtx x_a^*(t) \\ w^*(t) \esmtx^\top
    \bsmtx Q_i(t) & S_i(t) \\ S_i(t) ^\top & R_i(t) \esmtx
    \bsmtx x_a^*(t) \\ w^*(t) \esmtx \, dt 
  \end{align}
\end{enumerate}
\end{thm}
\begin{proof}
  By Theorem~\ref{thm:Jeval}, if $\lambda \in \Jdom$ then there exists
  a solution $Y$ to the RDE with boundary condition $Y(T)=0$.
  Moreover,
  $J(\lambda)=\bsmtx 0 \\ 1\esmtx^\top Y(0) \bsmtx 0 \\ 1 \esmtx$.  It
  is a standard result in LQ optimal control that the optimal cost to
  \eqref{eq:LQcost} is
  $\bsmtx 0 \\ 1\esmtx^\top Y(0) \bsmtx 0 \\ 1 \esmtx$ achieved by
  $(x_a^*,w^*)$ as defined in Statement 2. See Chapter 2 of
  \cite{anderson07} or Proposition 8 of \cite{jonsson02}.

  Next, consider Statement 3.  If $\alpha \notin \Jdom$ then
  $J(\alpha)=\infty$ and hence \eqref{eq:Jsubgrad} holds trivially.
  Thus consider $\alpha \in \Jdom$. The pair $(x_a^*,w^*)$ is optimal
  for the LQ optimization defined with $\lambda$.  It provides
  a lower bound on the maximal cost of the LQ optimization defined
  by any other $\alpha \in \Jdom$:
  \begin{align}
    \label{eq:Jsubgrad2}
    J(\alpha) \ge 
    \int_0^T \bsmtx x_a^*(t) \\ w^*(t) \esmtx^\top
    \bsmtx Q(t,\alpha) & S(t,\alpha) \\ 
    S(t,\alpha)^\top & R(t,\alpha) \esmtx
    \bsmtx x_a^*(t) \\ w^*(t) \esmtx \, dt 
  \end{align}
  Finally, it follows from 
  the definition of $Q$ in \eqref{eq:QSR} that
  $Q(t,\alpha)=Q(t,\lambda)+\sum_{i=1}^m (\alpha_i-\lambda_i)Q_i$.
  Similar relationships hold for $S$ and $R$. Thus
  \eqref{eq:Jsubgrad2} can be equivalently written
  as \eqref{eq:Jsubgrad}.
\end{proof}

By Theorem~\ref{thm:Jsubgrad}, a subgradient for $J$ at
$\lambda \in \Jdom$ can be evaluated using the solution $Y$ of the
RDE. First, the signals $(x_a^*,w^*)$ are obtained by solving the
dynamics \eqref{eq:xawstar} from the initial condition
$x_a^*(0) =\bsmtx 0 \\ 1 \esmtx$. Second, the subgradient $g$ is then
obtained by performing the integrals in \eqref{eq:gi}. These
two steps have a small computational cost relative to cost of solving
the RDE itself.

If $\lambda \notin \Jdom$ then we can also construct a $g \in \R^m$
that separates $\lambda$ from the feasible set
$\Jdom$~\cite{jonsson02,kao01}. The construction is summarized here.
If $\lambda \notin \Jdom$ then the RDE does not have a solution on
$[0,T]$. Specifically, the solution $Y$ to $RDE(Y,t,\lambda)=0$ grows
unbounded when integrated backward from $Y(T)=0$. Thus the solution
exists only $(t_0,T]$ for some $t_0\in (0,T)$.  In this case, there
exists non-trivial signals $(x_a^*,w^*)$ on $[t_0,T]$ that satisfy:
\begin{align}
  \label{eq:xawstarzeroIC}
\begin{split}
  &   \int_{t_0}^T \bsmtx x_a^*(t) \\ w^*(t) \esmtx^\top
  \bsmtx Q(t,\lambda) & S(t,\lambda) \\ 
  S(t,\lambda)^\top & R(t,\lambda) \esmtx
  \bsmtx x_a^*(t) \\ w^*(t) \esmtx \, dt  = 0 \\
  & \dot{x}_a^* = A_a x_a^* + B_a w^* 
    \, \mbox{ with }  x_a(t_0)=0 
\end{split}
\end{align}
A numerical implementation for this construction is given in
\cite{iannelli18}.\footnote{Briefly, integrate $Y$ backward to
  $Y(t_0+\epsilon)$ for some sufficiently small $\epsilon>0$.  Let
  $\rho_\epsilon$ be the spectral radius of $Y(t_0+\epsilon)$ and note
  that $\rho_\epsilon\to \infty$ as $\epsilon\to 0$.  Let $v_\epsilon$
  be the corresponding eigenvector of $Y(t_0+\epsilon)$ associated
  with $\rho_\epsilon$. Solve the dynamics \eqref{eq:xawstar} with
  initial condition $x_a^*(0)=\rho_\epsilon^{-1} v_\epsilon$ to obtain
  $(x_a^*,w^*)$.}  We can use this pair $(x_a^*,w^*)$ to construct $g$
from \eqref{eq:gi}. By linearity, this pair satisfies the following
for any $\alpha \in \R^m$:
\begin{align*}
  & \int_{t_0}^T \bsmtx x_a^*(t) \\ w^*(t) \esmtx^\top
  \bsmtx Q(t,\alpha) & S(t,\alpha) \\ 
  S(t,\alpha)^\top & R(t,\alpha) \esmtx
    \bsmtx x_a^*(t) \\ w^*(t) \esmtx \, dt  \\
   &  \hspace{1in} = g^\top (\alpha - \lambda)
\end{align*}
Thus if $g^\top (\alpha-\lambda)\ge0$ then $(x_a^*,w^*)$ are non-trivial
signals that yield non-negative LQ cost from $x_a^*(0)=0$. The RDE
fails to exist in this case by the LTV Bounded Real Lemma (Theorem 1
in \cite{seiler19}). Hence the feasible set satisfies
$\Jdom \subset \{ \alpha \, : \, g^\top (\alpha - \lambda)<0 \}$ In
other words, $g$ defines a hyperplane that separates
$\lambda \notin \Jdom$ from the feasible set $\Jdom$.

\subsection{Solution Via Ellipsoidal Algorithm}
\label{sec:ellipalg}

Assume the convex optimization \eqref{eq:minJ},
$\min_{\lambda\in \Jdom,\lambda\ge 0} J(\lambda)$, is feasible with
optimal point $\lambda^*$ and optimal cost $\Jopt=J(\lambda^*)$.
Algorithm~\ref{alg:EA} provides pseudo-code to solve this using the
ellipsoidal algorithm (Section 14.4 of \cite{boyd91}).


The algorithm computes a sequence of ellipsoids defined with a center
$\lambda \in \R^n$ and shape matrix $\Lambda \succ 0$ as follows:
\begin{align*}
\E(\Lambda,\lambda) := \{ \alpha \in \R^n \, : \,
    (\alpha - \lambda)^\top \Lambda^{-1} (\alpha - \lambda) \le 1 \}
\end{align*}
The algorithm is initialized with a sphere of radius $R>0$ centered at
the origin: $\Lambda^{(0)} = R^2 \cdot I$ and $\lambda^{(0)} =0$.
Assume that $\lambda^* \in \Lambda^{(0)}$. 

For each step $k=0,1,\ldots$, the cost is evaluated at the ellipsoid
center $J(\lambdak)$. Moreover, a vector $\gk\in \R^n$ is
computed that is either a subgradient at $\lambdak$, if feasible,
or separates $\lambdak$ from the infeasible set. Additional
details on the computation of $\gk$ are given below.  The optimal
point lies in the intersection of $\E(\Lambdak,\lambdak)$
and the half space
$\mathcal{H}(\gk,\lambdak):=\{ \alpha \, : \, (\gk)^\top
(\alpha - \lambdak)<0 \}$. 

The ellipsoid algorithm computes the smallest ellipsoid that contains
$\E(\Lambdak,\lambdak) \cap
\mathcal{H}(\gk,\lambdak)$.  The updated ellipsoid is defined
in terms of the normalized vector
$\tilde{g} := \frac{1}{\sqrt{(\gk)^\top A \gk}} \gk$:
\begin{align}
\label{eq:EMupdate}
\lambda^{(k+1)} & := \lambdak - \frac{1}{m+1} \Lambdak \tilde{g} \\
\nonumber
\Lambda^{(k+1)} & :=\frac{m^2}{m^2-1} \left( \Lambdak
   - \frac{2}{m+1} \Lambdak \tilde{g} \tilde{g}^\top \Lambdak \right)
\end{align}
If $\lambda^* \in \E(\Lambda^{(0)},\lambda^{(0)})$ then the optimal
point remains in the ellipsoid at each iteration. In addition, if
$\lambdak$ is feasible then the optimality gap is bounded by:
\begin{align}
  \label{eq:Jbnd}
  \begin{split}
    &  J(\lambdak) - \Jopt \le  J_{bnd} \\
    & \mbox{where } J_{bnd}:= \sqrt{ (\gk)^\top \Lambdak \gk}
  \end{split}
\end{align}
The ellipsoids shrink by a fixed factor $e^{-\frac{1}{2m}}$ at each
iteration.  If the optimization \eqref{eq:minJ} is feasible and
$\lambda^* \in \E(\Lambda^{(0)},\lambda^{(0)})$ then the algorithm
converges $J(\lambdak) \to \Jopt$.  Details for these
facts are given in Section 14.4 of \cite{boyd91}.

Finally, we discuss the construction of $\gk$ based on 
three possible cases:
\begin{itemize}
\item \emph{Case 1:} The current ellipsoid center $\lambdak$ has at
  least one negative entry, $\lambda_i^{(k)}<0$ for some $i$. Thus the
  ellipsoid center is infeasible as it does not satisfy
  $\lambdak \ge 0$. In this case $J(\lambdak)=+\infty$ and any
  feasible point must satisfy $\alpha_i \ge 0 > \lambda_i^{(k)}$.  Set
  $\gk=-e_i$ where $e_i\in \R^m$ is the $i^{th}$ basis vector.  Thus
  the feasible set lies in the half space $\mathcal{H}(\gk,\lambdak)$.
\item \emph{Case 2:} If the current ellipsoid center satisfies
  $\lambdak \ge 0$ then integrate the RDE backward from $Y(T)=0$.  If
  the solution exists on $[0,T]$ then
  $J(\lambdak) = \bsmtx 0 \\ 1 \esmtx^\top Y(0) \bsmtx 0 \\ 1 \esmtx$.
  Moreover, a subgradient $\gk$ can be computed from
  Theorem~\ref{thm:Jsubgrad}.  The set of feasible points with
  strictly lower cost lies in $\mathcal{H}(\gk,\lambdak)$.
\item \emph{Case 3:} Suppose the current ellipsoid center satisfies
  $\lambdak \ge 0$ but the RDE fails to exist on $[0,T]$. Then the
  ellipsoid center is infeasible because $\lambdak \notin \Jdom$.  As
  discussed in Section~\ref{sec:subgrad}, a separating hyperplane
  $\gk$ can be constructed from \eqref{eq:gi} and a pair $(x_a^*,w^*)$
  that satisfy \eqref{eq:xawstarzeroIC}. The feasible set lies in the
  half space $\mathcal{H}(\gk,\lambdak)$.
\end{itemize}



\begin{algorithm}
  \linespread{1}\selectfont
  \caption{Ellipsoidal Algorithm} \label{alg:EA}  
  \begin{algorithmic}[1]
    \State \textbf{Given:} $G_a$, $\epsilon_{tol}>0$,  $R>0$
    \State \textbf{Initialize:} $\Lambda^{(0)}=R^2\cdot I$, $\lambda^{(0)}=0$,
      $\Jbnd=\infty$, $k=0$    
    \While{$\Jbnd>\epsilon_{tol}$}
    \If{$\lambda_i^{(k)}<0$ for some $i$} 
    \State $J(\lambdak)=\infty$
    \State $\gk=-e_i$ \, ($e_i$ is $i^{th}$ basis vector).
    \Else
    \State Solve $RDE(t,Y,\lambda)=0$; $Y(T)=0$.
    \If{$Y$ exists on $[0,T]$}
    \State $J(\lambdak) = \bsmtx 0 \\ 1\esmtx^\top Y(0) \bsmtx 0 \\ 1\esmtx$
    \State Compute $\gk$ from Theorem~\ref{thm:Jsubgrad}.
    \Else
    \State $J(\lambdak)=\infty$
    \State Compute $\gk$ from \eqref{eq:xawstarzeroIC} and \eqref{eq:gi}
    \EndIf
    \EndIf 
    \State Sub-optimality:
    $\Jbnd:=\sqrt{ (\gk)^\top \Lambdak \gk}$

    \State Compute $(\Lambda^{(k+1)},\lambda^{(k+1)}$) from \eqref{eq:EMupdate}
    \State Update iteration number: $k=k+1$
    \EndWhile
  \end{algorithmic}
\end{algorithm}

\section{Example}
\label{sec:ex}

This example considers the robustness of a two link robot arm
(Figure~\ref{fig:twoLinkRobot}) as it traverses a finite-time
trajectory. This example was previously used to study related
finite-horizon robustness issues in
\cite{seiler19,seiler17arXiv,moore15}.  The mass and moment of inertia
of the $i$-th link are denoted by $m_i$ and $I_i$.  The robot
properties are $m_1=3kg$, $m_2 = 2kg$, $l_1 = l_2= 0.3m$,
$r_1=r_2 = 0.15m$, $I_1= 0.09 kg\cdot m^2$, and
$I_2= 0.06 kg\cdot m^2$.  The nonlinear equations of motion (Section
2.3 of \cite{murray94}) are given by:
\begin{align}
\label{eq:linkDyns}
&  \bmat{\alpha+ 2\beta\cos(\theta_2) & \delta + \beta \cos(\theta_2) \\
    \delta +  \beta \cos(\theta_2) & \delta}
  \bmat{\ddot{\theta}_1 \\ \ddot{\theta}_2} + \\
\nonumber
&  \bmat{-\beta \sin(\theta_2) \dot{\theta}_2 &
    -\beta \sin(\theta_2) (\dot{\theta}_1 + \dot{\theta}_2) \\
    \beta \sin(\theta_2) \dot{\theta}_1 & 0} \bmat{\dot{\theta}_1 \\
    \dot{\theta}_2} =\bmat{\tau_1 \\ \tau_2}  
\end{align}
with
\begin{align*}
& \alpha := I_{1} + I_{2} + m_1r_1^2 
+ m_2(l_1^2 + r_2^2) = 0.443    \, kg \cdot m^2\\
& \beta := m_2 l_1 r_2 = 0.09 \, kg \cdot m^2 \\
& \delta := I_{2} + m_2 r_2^2 = 0.105 \, kg \cdot m^2.
\end{align*}
The state and input are
$\eta :=[\theta_1  \ \theta_2  \ \dot{\theta}_1 \
  \dot{\theta}_2]^T$ and $\tau :=[\tau_1 \ \tau_2]^T$,
  where $\tau_i$ is the torque applied to the base of link $i$.  A
trajectory $\bar{\eta}$ was selected for the arm and the required
input torque $\bar{\tau}$ was computed.
Figure~\ref{fig:LinkRobotFigure2} shows the desired trajectory for the
tip of arm two (red dashed line) in Cartesian coordinates from $t=0$
to $T=5$ sec.  The arm positions at four different times are
also shown. 

\begin{figure}[t] 
  \centering
  \includegraphics[scale=0.3]{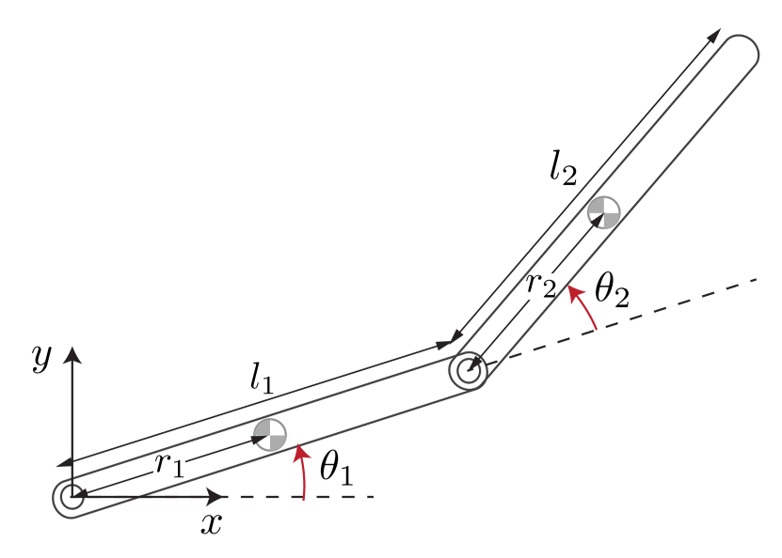}
  \caption{Two link robot arm \cite{murray94}.}
  \label{fig:twoLinkRobot}
\end{figure}

\begin{figure}[t] 
  \centering
  \includegraphics[scale=0.5]{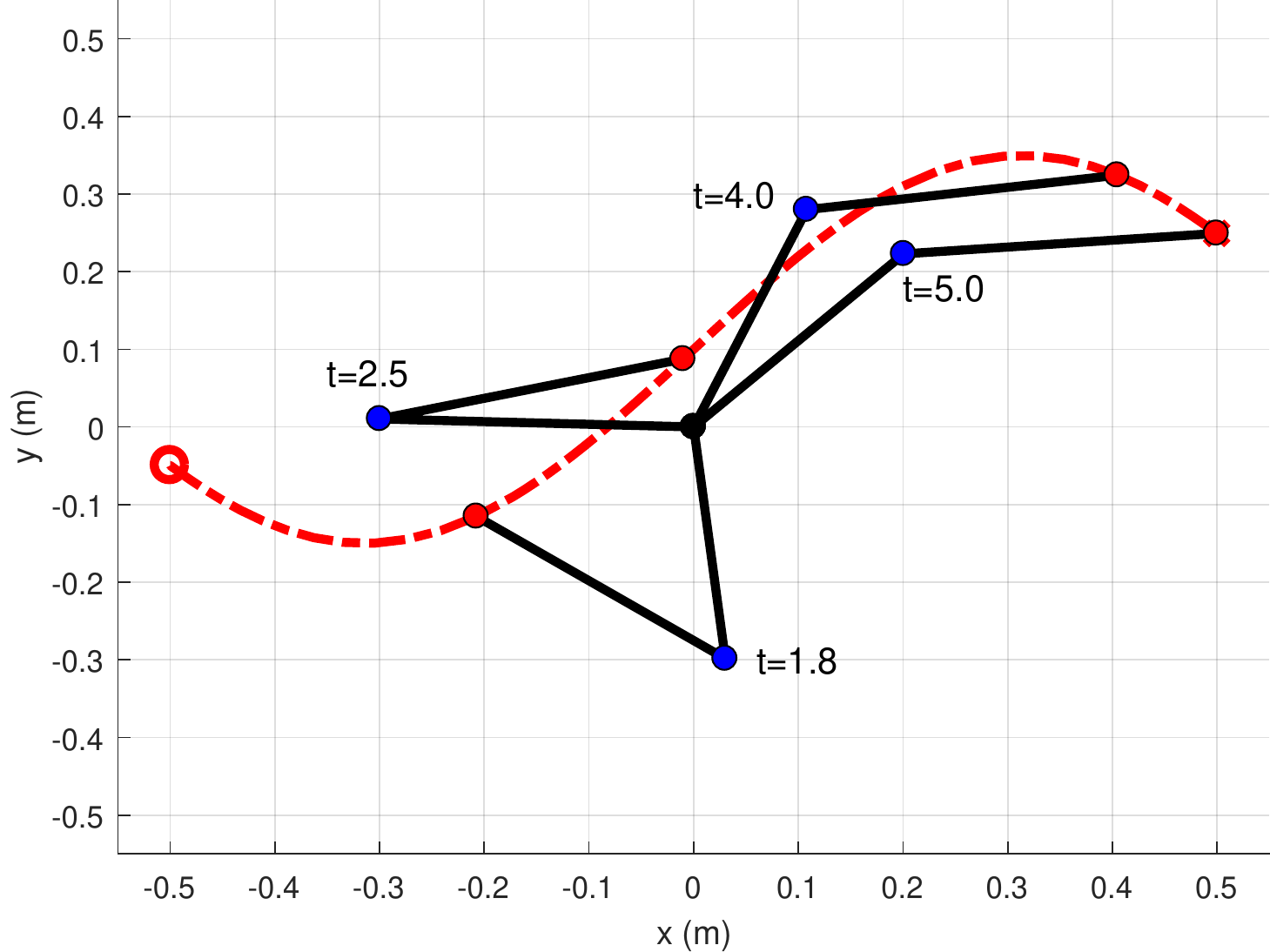}
  \caption{Desired trajectory in Cartesian coordinates
    (dotted red line) and robot arm position at four times.}
  \label{fig:LinkRobotFigure2}
\end{figure}

A state feedback law is implemented to track this trajectory.  The
input torque vector is $\tau = \bar{\tau} + u$ where
$u(t)=K(t)\, (\bar\eta(t) - \eta(t))$.  The feedback gain is
constructed via finite horizon, LQR design. Details on the trim
trajectory and state feedback design can be found in
\cite{seiler17arXiv,moore15}. Figure~\ref{fig:UncRobot} shows a block
diagram for the uncertain, nonlinear dynamics for the two-link robot
and state feedback.   The analysis aims to bound the $\mathcal{L}_2[0,T]$
norm of the tracking error $e(t)=\bar\eta(t) - \eta(t)$ in the presence
of uncertainty in the joint torques.

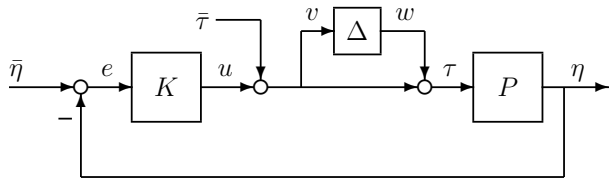
\begin{figure}[t]
\centering
\scalebox{0.85}{
  \begin{picture}(267,95)(3,-35)
    \thicklines
    \put(3,0){\vector(1,0){29}}
    \put(3,5){$\bar\eta$}
    \put(35,0){\circle{6}}
    \put(38,0){\vector(1,0){20}}
    \put(44,5){$e$}
    \put(58,-15){\framebox(30,30){$K$}}
    \put(96,5){$u$}
    \put(88,0){\vector(1,0){24}}
    \put(115,0){\circle{6}}
    \put(115,30){\vector(0,-1){27}}
    \put(95,30){\line(1,0){20}}
    \put(86,26){$\bar\tau$}
    \put(118,0){\vector(1,0){67}}
    \put(188,0){\circle{6}}
    \put(133,0){\line(0,1){25}}
    \put(133,25){\vector(1,0){15}}
    \put(135,30){$v$}
    \put(148,15){\framebox(20,20){$\Delta$}}
    \put(175,30){$w$}
    \put(168,25){\line(1,0){20}}
    \put(188,25){\vector(0,-1){22}}
    \put(191,0){\vector(1,0){19}}
    \put(196,5){$\tau$}
    \put(210,-15){\framebox(30,30){$P$}}
    \put(240,0){\vector(1,0){30}}
    \put(253,5){$\eta$}
    \put(250,0){\line(0,-1){40}}
    \put(250,-40){\line(-1,0){215}}
    \put(35,-40){\vector(0,1){37}}
    \put(25,-14){\line(1,0){6}}
\end{picture}
} 
\caption{Uncertain Nonlinear Model for Two-Link Robot}
\label{fig:UncRobot}
\end{figure}

Algorithm~\ref{alg:EA} was used to compute bounds on the
$\mathcal{L}_2[0,T]$ norm of the tracking error $e$. The causal,
uncertainty is assumed to be a full (2-by-2) with bounded induced
$\mathcal{L}_2[0,T]$ norm: $\| \Delta \|_{2\to 2} \le \beta$. The corresponding
IQC for this uncertainty is defined by
$M=\bsmtx \beta^2 I & 0 \\ 0 & -I \esmtx$. The IQC variable $\lambda$
is a scalar and hence Algorithm~\ref{alg:EA} reduces to bisection in
this case.\footnote{The ellipsoids simplify to intervals for scalar
  variables.  The interval is bisected based on the sign of the
  subgradient evaluated at the interval center.}
Algorithm~\ref{alg:EA} was run for uncertainty levels
$\beta = 0.05,0.1,\ldots,0.4$. The bisection was initialized with the
interval $[0,10]$ and was run until the optimality gap was less than
1\%.  The results are shown by the blue curve in
Figure~\ref{fig:JVsbeta1IQC}.  This took $\approx 290$sec for the
eight calculations on a standard laptop computer.

\begin{figure}[t] 
  \centering
  \includegraphics[scale=0.5]{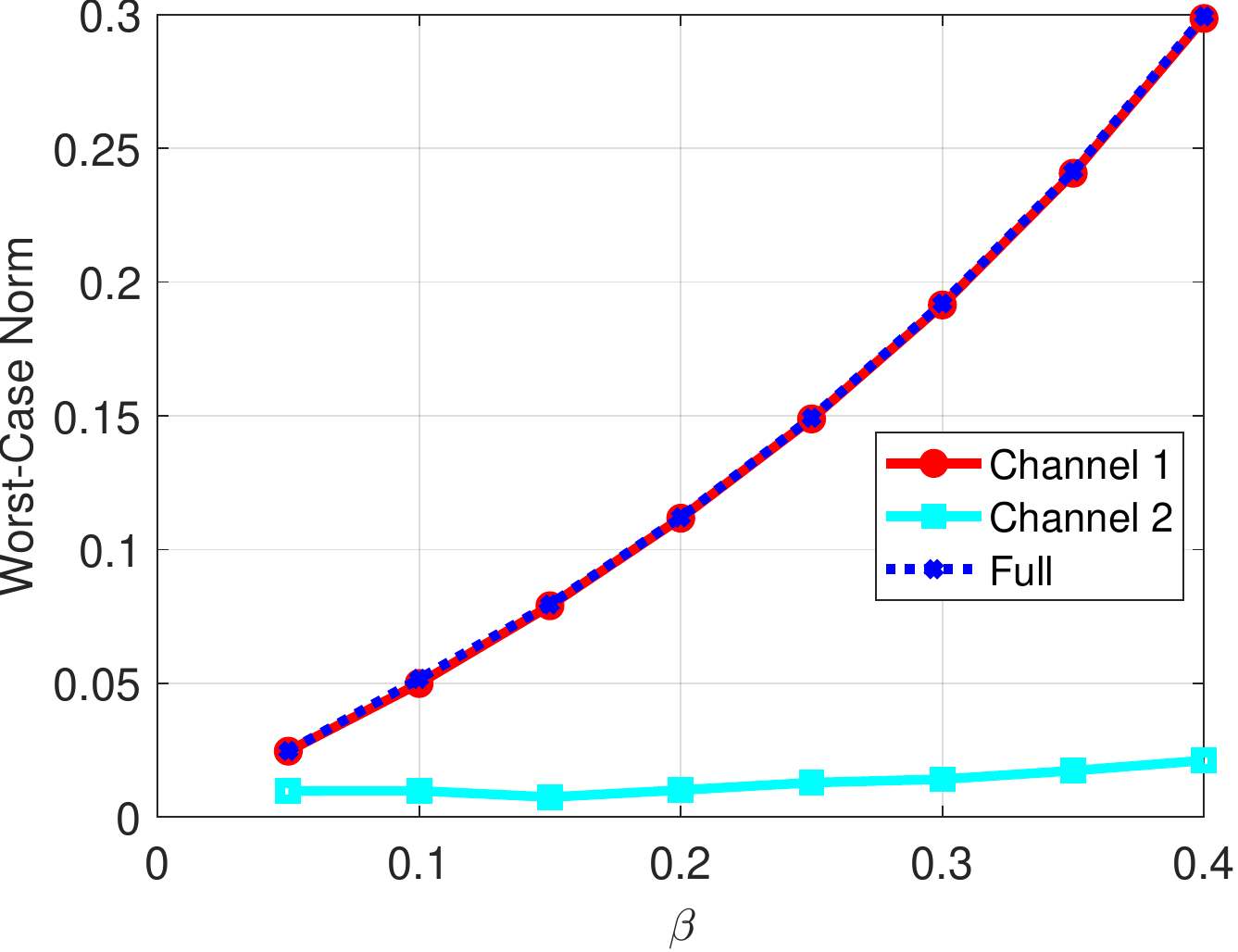}
  \caption{Worst-case norm versus uncertainty level $\beta$ for
    uncertainty structure that is full 2-by-2, on Channel 1 only,
    and Channel 2 only.}
  \label{fig:JVsbeta1IQC}
\end{figure}

The analysis was repeated under two additional assumptions on the
uncertainty: (i) uncertainty on channel 1,
$\Delta := \bsmtx \Delta_1 & 0 \\ 0 & 0 \esmtx$ with
$\| \Delta_1 \|_{2\to 2} \le \beta_1$, and (ii) uncertainty on channel
2, i.e.  $\Delta := \bsmtx 0 & 0 \\ 0 & \Delta_2 \esmtx$ with
$\| \Delta_2 \|_{2\to 2} \le \beta_2$.  Define
$E_1=\bsmtx 1 & 0 \\ 0&0\esmtx$ and $E_2=\bsmtx 0 & 0 \\ 0 &1\esmtx$ .
The corresponding IQCs for these two uncertainty sets are:
\begin{align}
\label{eq:M12} 
\mbox{(i)} \,
M_1=\bsmtx \beta_1^2 E_1 & 0 \\ 0 & -E_1 \esmtx, 
\, \mbox{(ii)} \,
M_2=\bsmtx \beta_2^2 E_2 & 0 \\ 0 & -E_2 \esmtx. 
\end{align}
Algorithm~\ref{alg:EA} was run for these two cases with same
uncertainty levels $\beta_i = 0.05,0.1,\ldots,0.4$ for $i=1,2$. The
algorithms again reduces to bisection as each case only has one IQC
variable.  The bisection was initialized with the interval $[0,10]$
and was run until the optimality gap was less than 1\%.  The results
are shown by the red and cyan curves in Figure~\ref{fig:JVsbeta1IQC}.
Cases (i) and (ii) took $\approx 240$sec and $\approx 200$sec,
respectively. A key observation from these results is that the results
for case (i) are very close to those for the full block uncertainty
while (ii) results in much smaller norm. This indicates that
uncertainty at the base of link 1 is more significant and uncertainty
at the base of link 2 has a negligible effect.

Finally, the analysis was performed with uncertainty in both channels
1 and 2 but no cross-coupling.  This corresponds to
$\Delta := \bsmtx \Delta_{1} & 0 \\ 0 & \Delta_2 \esmtx$ with
$\|\Delta_i\|_{2\to 2} \le \beta_i$ for $i=1,2$. This case has two
IQCs defined by $M_1$ and $M_2$ in \eqref{eq:M12} with the
corresponding $\beta_i$.  The analysis was performed with
$\beta_1=0.05$ and $\beta_2=0.8$, i.e.  much larger uncertainty in
channel 2.  Algorithm~\ref{alg:EA} was initialized with a sphere of
radius $R=20$. The ellipsoid algorithm was run to an optimality gap of
1\%.  This converged after 51 iterations and took $\approx 95$sec.  It
converged to $\lambda_1^*=0.0142$ and $\lambda_2^*=0.0447$ with
$\sqrt{\Jopt} =0.139$. 

For comparison, the original infinite-dimensional
formulation~\eqref{eq:DLMIOptim} was approximated by: (i) enforcing
the DLMI on a grid of 20 evenly spaced time points between $[0,5]$,
(ii) treating $P(t)$ at 10 evenly spaced time points as decision
variables, and (iii) using cubic splines to evaluate $P(t)$ and
$\dot{P}(t)$ at the DLMI grid points. Solving this (approximate)
finite-dimensional SDP gives $\lambda_{1,SDP}=0.0195$ and
$\lambda_{2,SDP}=0.0448$ with $\sqrt{\Jopt} \approx 0.147$. We then
solved the RDE using these IQC variables and obtained a bound of
$\sqrt{\Jopt} \le \sqrt{J(\lambda_{SDP})} = 0.140$. Solving the SDP
followed by the RDE only took $\approx 6$secs. Thus the standard
heuristic (gridding and basis functions) yields a nearly optimal
answer with much less computation time on this particular example.
However, no optimality gap is provided with this heuristic.

\section{Conclusions}

This paper presented a method to analyze the robustness of an
uncertain nonlinear system along a finite-horizon trajectory.  The
approach relies on a linearization of the nominal nonlinear system
along the trajectory. A DLMI condition was then developed using the
linearized, uncertain LTV system.  This led to an infinite dimensional
convex optimization to assess robustness.  This optimization was then
converted to an equivalent finite dimensional optimization based on a
related Riccati Differential Equation.  The ellipsoidal method to
solve this optimization avoids heuristics often used to solve DLMIs,
e.g. time gridding. The approach was demonstrated by a two-link
robotic arm example. This included a comparison of the ellipsoid
method and heuristic gridding approaches. Future work will include
implementations for more general (dynamic, soft) IQCs.

\bibliographystyle{IEEEtran}
\bibliography{ltvbib}

\appendix
\section{Lemmas on RDEs}

The main text denoted the Riccati Differential Equation by
$RDE(t,Y)=0$ (dropping any dependence on $\lambda$). This appendix
focuses on the case where the right side is a matrix function
of time $W:[0,T] \to \Sm^{n_x+1}$. Specifically, $RDE(t,Y)=W$ denotes
the following RDE:
\begin{align}
\label{eqApp:RDEW}
\begin{split}
& \dot{Y} + A_a^\top Y + Y A_a + Q \\
& \hspace{0.2in} - (Y B_a + S) R^{-1} (Y B_a + S)^\top = W
\end{split}
\end{align}
This appendix provides supporting lemmas regarding the monotonicity
and continuity of the RDE solutions under such perturbations.

\begin{lem}
  \label{lemApp:RDEmon}
  Assume:
  \begin{enumerate}
  \item $R(t)\prec 0$ for all $t \in [0,T]$
  \item $W_i:[0,T]\to \Sm^{n_x+1}$ are functions $(i=1,2)$ that
    satisfy $W_1(t) \preceq W_2(t)$ $\forall t\in [0,T]$.
  \item $Y_i:[0,T]\to \Sm^{n_x+1}$ are differentiable functions
    ($i=1,2$) that satisfy $Y_1(T) \succeq Y_2(T)$ and
    $RDE(t,Y_i)=W_i$ $\forall t \in [0,T]$.
  \end{enumerate}
  Then $Y_1(t) \succeq Y_2(t)$ $\forall t\in [0,T]$.
\end{lem}
\begin{proof}
  Define the difference as $E:=Y_1-Y_2$ and note that $E(T)\succeq
  0$.  Subtract $RDE(t,Y_2)=W_2$ from $RDE(t,Y_1)=W_1$
  to obtain, after some algebra, the following expression:
  \begin{align*}
    \dot{E} + \hat{A}^\top E + E \hat{A} = F
  \end{align*}
  where
  \begin{align*}
    \hat{A}& :=A_a-B_aR^{-1}(Y_2B_a+S)^\top, \\
    F & := W_1-W_2 +  E B_a R^{-1} B_a^\top E.
  \end{align*}
  The assumptions on $W_1$, $W_2$, and $R$ imply that $F(t)\preceq 0$
  for all $t\in [0,T]$. Next, let $\Phi(t,\tau)$ be the state
  transition matrix associated with $\dot{x}(t) = -\hat{A}^\top(t) x(t)$. By
  Lemma 10.3 in \cite{bitmead12}, the solution $E$ can be expressed
  as:
  \begin{align}
    \label{eqApp:Esol}
     E(t) =&  \Phi(t,T) E(T) \Phi(t,T)^\top \\
    \nonumber
     & - \int_t^T \Phi(t,\tau) F(\tau)  \Phi(t,\tau)^\top  \, d\tau
  \end{align}
  Equation~\ref{eqApp:Esol} implies that $E(t)=P(t)-Y(t)\succeq 0$
  because $E(T)\succeq 0$ and $F(t) \preceq 0$ $\forall t\in [0,T]$.
  Hence $Y_2(t)\succeq Y_1(t)$ $\forall t\in [0,T]$.
\end{proof}

\begin{lem}
  \label{lemApp:RDEcont}
  Assume $(A_a,B_a,Q,R,S)$ are all continuous functions of time.
  Moreover, assume:
  \begin{enumerate}
  \item $R(t)\prec 0$ for all $t \in [0,T]$
  \item $Y_0:[0,T]\to \Sm^{n_x+1}$ is a differentiable function
    satisfying $Y_0(T) =0$ and $RDE(t,Y_0)=0$ $\forall t \in [0,T]$.
  \item $\{ \epsilon_k \}_{k=1}^\infty \subset \R$ are positive scalars
    satisfying $\epsilon_k \ge \epsilon_{k+1}$ $\forall k$ and
    $\lim_{k\to \infty} \epsilon_k = 0$.
  \item $Y_k:[0,T]\to \Sm^{n_x+1}$ are differentiable functions
    ($k=1,2,\ldots$) that satisfy $Y_k(T) =0$ and
    $RDE(t,Y_k)=-\epsilon_k \cdot I$ $\forall t \in [0,T]$.
  \end{enumerate}
  Then $\lim_{k\to \infty} \|Y_k(t)- Y_0(t)\|=0$ $\forall t\in [0,T]$.
\end{lem}
\begin{proof}
  The proof adapts similar arguments used to prove Theorem 3 in
  \cite{czornik00}. First, the horizon is finite $T<\infty$ and
  $[0,T]$ is a compact set.  Hence the solution $Y_1$ is uniformly
  bounded, i.e. there exists $c_1<\infty$ such that:
  \begin{align}
    \max_{t \in [0,T]} \| Y_1(t) \| \le c_1
  \end{align}
  where $\|Y_1(t)\|$ is the (matrix) induced $2$-norm of $Y_1(t)$.  By
  Lemma~\ref{lemApp:RDEmon}, $Y_1(t)\succeq Y_0(t)$ and
  $Y_1(t) \succeq Y_k(t)$ ($k=2,3,\ldots$) for all $t\in [0,T]$.  Thus
  $Y_0$ and $\{Y_k\}_{k=2}^\infty$ are also uniformly bounded by
  $c_1$.

  Next, define $E_k:=Y_k-Y_0$ and note that $E_k(T)=0$. Moreover,
  define:
  \begin{align*}
    \hat{A} & :=A_a-B_aR^{-1}(Y_0B_a+S)^\top,    \\
    F_k & := -\epsilon_k I +  E_k B_a R^{-1} B_a^\top E_k.
  \end{align*}
  Following the same arguments as in the proof of
  Lemma~\ref{lemApp:RDEmon}, the solution $E_k$ can be expressed as:
  \begin{align}
    \label{eqApp:Eksol}
    E_k(t) =  - \int_t^T \Phi(t,\tau) F(\tau)  \Phi(t,\tau)^\top  \, d\tau,
  \end{align}
  where $\Phi(t,\tau)$ is the state transition matrix for
  $-\hat{A}^\top$. The state transition matrix is also uniformly bounded:
  $\exists c_2$ such that $\|\Phi(t,\tau)\|\le c_2$
  $\forall t,\tau \in [0,T]$. Moreover, $\exists c_3$ such that such
  that $\|B_a(t)R^{-1}(t) B_a(t)\|\le c_3$ for all $t\in [0,T]$.  Use these
  constants and \eqref{eqApp:Eksol} to bound $E_k$ as follows:
  \begin{align*}
    \| E_k(t)\| & \le  \int_t^T c_2^2 \left( \epsilon_k 
                 +  c_3 \|E_k(\tau)\|^2 \right)  \, d\tau \\
    & \le c_2^2\epsilon_k T 
      + \int_t^T   2c_1 c_2^2 c_3 \|E_k(\tau)\|  \, d\tau 
  \end{align*}
  where the second line follows from
  $\| E_k(\tau)\| \le \|Y_k(\tau)\| + \|Y_0(\tau)\|$ combined with the
  uniform bound $c_1$ defined above. Apply Gr\"{o}nwall's inequality
  (Lemma 1 in Section 5.7 of \cite{vidyasagar93}) to obtain:
  \begin{align}
    \|E_k(t)\| \le c_2^2\epsilon_k T \cdot
      e^{ 2c_1 c_2^2 c_3 (T-t) }
  \end{align}
  Finally, $\lim_{k\to \infty} \epsilon_k = 0$ implies that
  $\lim_{k\to \infty} \|E_k(t)\| =0$ for all $t\in [0,T]$.
\end{proof}

\end{document}